\def\year{2022}\relax
\documentclass[letterpaper]{article} 
\usepackage{aaai22}  
\usepackage{times}  
\usepackage{helvet}  
\usepackage{courier}  
\usepackage[hyphens]{url}  
\usepackage{graphicx} 
\usepackage{natbib}  
\usepackage{caption} 
\DeclareCaptionStyle{ruled}{labelfont=normalfont,labelsep=colon,strut=off} 
\usepackage{algorithm}
\usepackage{algorithmic}

\usepackage{newfloat}
\usepackage{listings}

\usepackage[switch]{lineno}
\makeatletter
\AtBeginDocument{%
	\@ifpackageloaded{amsmath}{%
		\newcommand*\patchAmsMathEnvironmentForLineno[1]{%
			\expandafter\let\csname old#1\expandafter\endcsname\csname #1\endcsname
			\expandafter\let\csname oldend#1\expandafter\endcsname\csname end#1\endcsname
			\renewenvironment{#1}%
			{\linenomath\csname old#1\endcsname}%
			{\csname oldend#1\endcsname\endlinenomath}%
		}%
		\newcommand*\patchBothAmsMathEnvironmentsForLineno[1]{%
			\patchAmsMathEnvironmentForLineno{#1}%
			\patchAmsMathEnvironmentForLineno{#1*}%
		}%
		\patchBothAmsMathEnvironmentsForLineno{equation}%
		\patchBothAmsMathEnvironmentsForLineno{align}%
		\patchBothAmsMathEnvironmentsForLineno{flalign}%
		\patchBothAmsMathEnvironmentsForLineno{alignat}%
		\patchBothAmsMathEnvironmentsForLineno{gather}%
		\patchBothAmsMathEnvironmentsForLineno{multline}%
	}{}
}
\makeatother

\floatstyle{ruled}
\newfloat{listing}{tb}{lst}{}
\floatname{listing}{Listing}
\usepackage{color} 

\usepackage{amsfonts}       
\usepackage{nicefrac}       
\usepackage{amssymb}
\usepackage{mathtools}
\usepackage{amsmath}
\usepackage{amsthm}
\usepackage{bm}
\usepackage{thm-restate}
\usepackage{todonotes}
\usepackage{thmtools}

\newtheorem{theorem}{Theorem}

\newcommand{\p}{\mathbf{p}}
\newcommand{\bp}{\mathbf{b}}
\newcommand{\bpi}{\bm{\pi}}

\newcommand{\SW}{\mathrm{SW}}
\newcommand{\Rev}{\mathrm{Rev}}
\newcommand{\ivcgstar}{${\cal M}_{\rm I}^{\rm VCG*}$}
\newcommand{\dvcg}{${\cal M}_{\rm D}^{\rm VCG}$}

\newcommand{\ivcg}{${\cal M}_{\rm I}^{\rm VCG}$}
\newcommand{\igsp}{${\cal M}_{\rm I}^{\rm GSP}$}

\newcommand{\tp}{\tilde{\p}}

\title{Efficiency of Ad Auctions with Price Displaying}

\author{
	Matteo Castiglioni\textsuperscript{\rm 1}, Diodato Ferraioli\textsuperscript{\rm 2}, Nicola Gatti\textsuperscript{\rm 1}, Alberto Marchesi\textsuperscript{\rm 1}, Giulia Romano\textsuperscript{\rm 1}\thanks{All the authors contributed equally.} \\
}	
\affiliations{
	\textsuperscript{\rm 1}Politecnico di Milano, Piazza Leonardo da Vinci 32, Milan, Italy \\
	\textsuperscript{\rm 2}Università degli Studi di Salerno, Via Giovanni Paolo II, Fisciano, Italy \\
	\{matteo.castiglioni, nicola.gatti, alberto.marchesi, giulia.romano\}@polimi.it, dferraioli@unisa.it
}

\usepackage{bibentry}

\begin{document}
\maketitle

\begin{abstract}
Most of the economic reports forecast that almost half of the worldwide market value unlocked by AI over the next decade (up to 6 trillion USD per year) will be in {\em marketing\&sales}.
In particular, AI will enable the optimization of more and more intricate economic settings, in which multiple different activities need to be jointly automated.
This is the case of, \emph{e.g.}, \emph{Google Hotel Ads} and \emph{Tripadvisor}, where auctions are used to display ads of similar products or services together with their prices.
As in classical ad auctions, the ads are ranked depending on the advertisers' bids, whereas, differently from classical settings, ads are displayed together with their prices, so as to provide a direct comparison among them.
This dramatically affects users' behavior, as well as the properties of ad auctions.
We show that, in such settings, social welfare maximization can be achieved by means of a direct-revelation mechanism that jointly optimizes, in polynomial time, the ads allocation and the advertisers' prices to be displayed with them.
%
However, in practice it is unlikely that advertisers allow the mechanism to choose prices on their behalf.
Indeed, in commonly-adopted mechanisms, ads allocation and price optimization are decoupled, so that the advertisers optimize prices and bids, while the mechanism does so for the allocation, once prices and bids are given.
We investigate how this decoupling affects the efficiency of mechanisms.
In particular, we study the \emph{Price of Anarchy} (PoA) and the \emph{Price of Stability} (PoS) of indirect-revelation mechanisms with both VCG and GSP payments, showing that the PoS for the revenue may be unbounded even with two slots, and the PoA for the social welfare may be as large as the number of slots.
Nevertheless, we show that, under some assumptions, simple modifications to the indirect-revelation mechanism with VCG payments achieve a PoS of $1$ for the revenue.
\end{abstract}

\section{Introduction}
Most of the economic reports forecast that \emph{artificial intelligence} (AI) will unlock up to 12 trillion USD per year worldwide by the next decade, and almost half of this amount will derive from the {\em marketing\&sales} area (see, \emph{e.g.},~\citep{McKinsey2018}). 
In particular, AI is playing a crucial role to tackle various problems, including, \emph{e.g.}, auction design~\citep{Bachrach2014Optimising}, the automation of advertisers' budget~\citep{Nuara2018Combinatorial} and bidding strategies~\citep{He2013Game}, and the optimization of conversion funnels~\citep{Nuara2019Dealing}. 

In this paper, we focus on recently-emerged online advertising settings where ad auctions are employed to display ads of similar products or services together with their prices. 
This is the case of, \emph{e.g.}, \emph{Google Hotel Ads} and \emph{Tripadvisor}, where users search for the availability of a hotel room in a given date. 
The web page of results shows a ranking of banners advertising similar hotel rooms that match the search criteria. 
Each banner displays the name of the advertiser providing the online booking service, together with the per-night selling price of the room.
Such settings are similar to standard ad auctions, since the ads are ranked depending on the advertisers' bids.
On the other hand, they also fundamentally differ from them, as the ad allocation must also take prices into account, and these are displayed inside the banners so as to provide a direct comparison among them.
%
%
This dramatically affects users' behavior, as well as the efficiency and the properties of ad auctions. 
The goal of this work is to investigate how the additional degree of freedom introduced by prices influences the problem of finding an optimal ad allocation and the revenue of the mechanisms.

The price-displaying feature of our setting introduces \emph{externalities among the ads}, since the probability that a user clicks on an ad depends on the prices displayed with both the ad being clicked and the other ads in the allocation.
Several forms of externalities are investigated in the literature on ad auctions.
However, to the best of our knowledge, no previous work takes into account price displaying. 
For instance, \citet{kempeWINE2008}~and~\citet{DBLP:conf/wine/AggarwalFMP08} introduce a basic user model that is currently adopted by most of the mechanisms.
In this model, a Markovian user observes the slots in a top-down fashion, moving down slot by slot with a given continuation probability and stopping on a slot to observe its ad with the remaining probability. 
\citet{kempeWINE2008} also propose richer models where the probability with which a user moves from a slot to the next one depends on the ad actually displayed in the former.
In this case, it is \emph{not} known whether the ad allocation problem admits a polynomial-time algorithm; however, \citet{DBLP:conf/aaai/FarinaG16,DBLP:journals/jair/FarinaG17} provide several algorithms
showing that in special cases a constant approximation can be achieved. Further externalities models are explored by \citet{fotakisSAGT2011} and \citet{GATTI2018150}, which allow for potentially different externalities for each pair of ads. 
However, with these models, the ad allocation problem is \textsf{NP}-hard and, in some cases, even inapproximable. 
It is also worth mentioning that similar models are adopted in mobile geo-located advertising by~\citet{DBLP:conf/aaai/GattiRCG14}.
In our model, we assume that the probability with which a user clicks on an ad depends on the price displayed with the ad \emph{and} on the lowest among all displayed prices. 
In particular, we model the click probability as a monotonically decreasing function of the ad price, assuming that the demand curve is monotonically decreasing in the price and that it is unlikely that a user clicks on an ad with a price larger than her reserve value.
We also assume that the click probability is monotonically decreasing in the difference between the ad price and the lowest displayed price, as the user's interest in any feature different from price (\emph{e.g.}, brand and loyalty) decreases as such difference increases.
%
%
%

In our setting, the private information of each advertiser (\emph{i.e.}, her type) is a pair composed by the probability with which a user visiting the advertiser's web page produces a conversion (\emph{e.g.}, a purchase) and the advertiser's cost for a unit of product or service.
On the other hand, the prices constitute an additional degree of freedom that can be controlled by either the advertisers or the mechanism.

As a first step, we present a direct-revelation mechanism that maximizes the social welfare by jointly optimizing over the ad allocations and the prices displayed with the ads.
Differently from what happens in most of the externalities models studied in the literature, such optimization problem can be solved in polynomial time for a given discretization of price values.
We also study the properties of the direct-revelation mechanism when VCG payments are used, showing that incentive compatibility, individual rationality and weak budget-balance hold in our setting. 
In real-world scenarios, it is unlikely that the advertisers let the mechanism select prices on their behalf, as required by the direct-revelation mechanism. 
In the (indirect-revelation) mechanisms that are currently adopted in real-world applications, the optimization over ad allocations and that over prices are decoupled.
In particular, each advertiser finds her optimal price and bid, while the mechanism optimizes over ad allocations once prices and bids are given. 
As for the direct-revelation mechanism, the best ad allocation can be found in polynomial time given prices and bids.
However, even if these indirect-revelation mechanisms allow the advertisers not to reveal private (and potentially sensitive) information, they can lead to inefficient equilibria.

We investigate the equilibrium inefficiency of indirect-revelation mechanisms with GSP and VCG payments, in terms of \emph{Price of Anarchy} (PoA) and \emph{Price of Stability} (PoS) in complete information settings.
In the literature, PoA and PoS are commonly-adopted efficiency metrics for standard ad auctions, in which the price variable is not taken into account.
For instance, \citet{caragiannisEC2011}, \citet{lucierEC2011}, and \citet{CARAGIANNIS2015343} show that the PoA for the social welfare of the GSP is upper bounded by $1.3$ with complete information and by $3$ with incomplete information, while \citet{DBLP:journals/jair/FarinaG17} and \citet{giotisWINE2008} study the inefficiency with specific externalities.
In our setting, the presence of externalities precludes the adoption of the tools provided by \citet{DBLP:journals/jair/RoughgardenST17} and \citet{hartlineEC2014} to bound the inefficiency of equilibria for the social welfare and the revenue, respectively, thus pushing us towards the development of \emph{ad hoc} approaches.
In particular, we show that, in our setting, the inefficiency of the indirect-revelation mechanisms with VCG and GSP mechanisms is much higher than that of the classical mechanisms without prices, even when excluding overbidding, since the PoS for the revenue may be unbounded even with two slots and the PoA for the social welfare may be as large as the number of slots.
Furthermore, with VCG payments, the PoS for the social welfare is $1$, while, with GSP payments, it is at least $2$, suggesting that GSP payments perform worse than VCG ones.

A crucial question is whether inefficiency can be reduced when letting the advertisers choose their prices.
We show that, under some assumptions, simple modifications to the indirect-revelation mechanism with VCG payments---requiring each advertiser to report an additional price---achieve a PoS of $1$ for the revenue.

\section{Formal Model}


%
There is a set $N=\{1,\dots,n\}$ of $n$ agents, who simultaneously  play the role of advertisers and sellers. 
%
Each agent sells a single good on her own website (\emph{e.g.}, an online marketplace) and relies on an external ad publisher that advertises the good through a single ad in which the price is displayed.
Since the goods being sold by the agents are similar, the price comparison that users perform on the publisher's website results in a high competition level among the agents, as happening in classical  comparator websites~\cite{JUNG20142079}.
%
In the following, for the ease of presentation, we use index $i \in N$ to refer to the agent, her good, and also her ad. 
Figure~\ref{figure:scenario} provides an overview of our scenario.

\begin{figure}[b!]
\begin{center}
\includegraphics[width=0.47\textwidth]{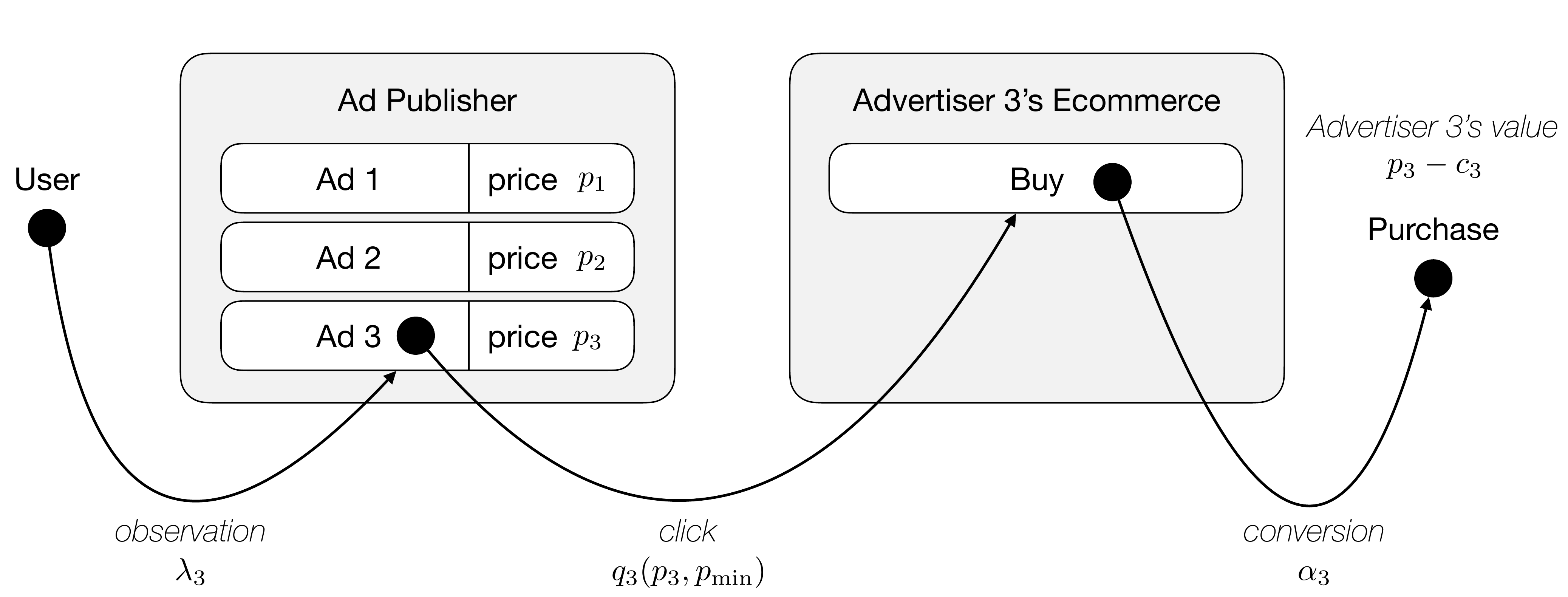}
\end{center}
\caption{An example of ad auction with price displaying. A user visits a web page with three ads (ad~$1$, ad~$2$, and ad~$3$) together with their prices ($p_1$, $p_2$, and $p_3$). The user observes slot~$3$ with probability $\lambda_3$. Once observed slot~$3$, the user clicks on the ad displayed in slot~$3$, \emph{i.e.}, ad $3$, with probability $q_3(p_3,p_{\min})$ where $p_{\min}$ is the minimum price among $p_1,p_2,p_3$. The user visits the web page of advertiser~$3$ (\emph{e.g.}, an online marketplace), and, then, produces a conversion (\emph{e.g.}, purchase) with probability $\alpha_3$. The value that advertiser~$3$ gets from the conversion is $p_3 - c_3$.}
\label{figure:scenario}
\end{figure}

%
For every $i$, we denote with $c_i \in \mathbb{R}_{\geq 0}$ and $p_i \in \mathbb{R}_{\geq 0}$ the \emph{cost} of supply and the selling \emph{price} of agent $i$'s good, respectively.
Furthermore, we denote with $\alpha_i \in [0,1]$ the probability with which a user buys agent $i$'s good when visiting her website.
Thus, agent~$i$'s expected \emph{gain} from a visit of a user on her website is $\alpha_i \, (p_i - c_i)$.
Let us remark that the \emph{conversion probability} $\alpha_i$ is constant w.r.t.~the price $p_i$, since we assume that the user is aware of the price before visiting the website and, thus, she does not visit it if the price is larger than her reserve value. 
As previously discussed, the user first observes the ads on the publisher's website, together with their prices, and, then, she clicks on an ad so as to visit the corresponding advertiser's website.
%
%
Therefore, the motivation behind an uncompleted conversion following the user's visit to the advertiser's web page does \emph{not} concern the price (\emph{e.g.}, it may be due to the user acquiring more information on the seller, or potential extra fees and/or ancillary services).
The pair $(\alpha_i,c_i)$ is a private information of agent~$i$, and sometimes we will refer to it as her type $\theta_i$. 
We let $\Theta = [0, 1] \times \mathbb{R}_{\geq 0}$ be the space of types of every agent.

%

%
The ad publisher has a set $M=\{1,\dots,m\}$ of slots in which the ads are displayed.
An \emph{assignment} of ads to slots (also called \emph{allocation}) is represented by a function $f \colon N \rightarrow M \cup \{\bot\}$ such that there is at most one ad per slot (\emph{i.e.}, there are no ads $i,h \in N$ such that $i\neq h$ and $f(i) = f(h) \in M$). 
All the ads that are not assigned to slots in $M$ are assigned to $\bot$, meaning that these ads are not displayed. 
For every slot $j \in M$, we denote with $\lambda_j \in [0, 1]$ the probability (called \emph{prominence}) that a user observes the ad displayed in that slot. 
As customary in the literature, we assume that $\lambda_1 \geq \lambda_2 \geq \ldots \geq \lambda_m$. 
For the ease of notation, we define $\lambda_\bot = 0$.
Furthermore, for every agent $i $, we denote with $q_i \in [0,1] $ the probability (called \emph{quality}) that a user clicks on ad~$i$ conditioned on its observation. 
In our setting, $q_i$ depends on the prices, as they are displayed with the ads. 
In particular, $q_i$ is a function of the prices $\p=\{p_i\}_{i\in N}$ of agents whose ads are displayed, since the user can compare all the prices shown on the web page when deciding the website from which to buy a good.
This dependency introduces externalities among the ads.
In this work, we assume that $q_i: \mathbb{R}_{\geq 0} \times \mathbb{R}_{\geq 0} \to [0,1]$, where $q_i(p_i,p_{\min})$ denotes the agent $i$'s quality when her price is $p_i$ and the minimum price among all the displayed ads is $p_{\min}$, with $p_{\min} = \min_{h \in N : f(h)\in M}\{p_h\}$ (for the sake of notation, we omit the dependency of $p_{\min}$ on $f$).
Moreover, given $p_{\min}$, $q_i$ is (non-strictly) monotonically decreasing in $p_i$ since, as previously discussed, a user clicks on the ad if the price is non-larger than the user's reserve value.
%
%
Finally, $q_i$ is (non-strictly) monotonically increasing in $p_{\min}$, given $p_i$.
%
The rationale behind this assumption is that, given $p_i$, the probability that a user clicks on ad~$i$ decreases as the gap between $p_i$ and the minimum price $p_{\min}$ increases, capturing a potential reduction of the user's interest for agent $i$'s good.
A simple example is when the users are only interested in the price and, thus, $q_i$ is zero if $p_i > p_{\min}$.
We also assume that there exists $p_i \in \mathbb{R}_{\geq 0}$ maximizing $q_i(p_i,p_i)\,\alpha_i\,(p_i-c_i)$ and, thus, there exists $p_i < \infty$ that agent~$i$ would use when displayed alone.
Finally, we remark that, as it is customary in the literature, parameters $\lambda$ and $q$ are estimated by the ad publisher.

Every \emph{mechanism} receives some input (or \emph{bid}) from every agent~$i$, chooses an allocation $f$, and charges every agent~$i$ of a payment $\pi_i$. 
We say that the mechanism is \emph{direct-revelation} if the input provided by agent $i$ belongs to $\Theta$, \emph{i.e.}, it consists of a conversion probability and a cost, which are \emph{not} necessarily the real ones (her type).
Otherwise we say that the mechanism is \emph{indirect-revelation}.


In our setting, a direct-revelation mechanism takes as input a reported type $ \theta'_i=(  \alpha'_i , c'_i) \in \Theta$ for each agent $i$, and chooses some prices $\p = \{ p_i\}_{i \in N}$  and an allocation function $f$.
We let $\bp = \{b_i\}_{i \in N}$ be the vector of declared gains, where $b_i= \alpha'_i\,(p_i -  c'_i)$ is agent $i$'s gain for the reported type $ \theta'_i$.
%
On the other hand, an indirect-revelation mechanism takes as input a price $p_i$ and a declared gain $b_i$ for each agent $i$, and chooses an allocation function $f$.
%
%
%
We say that agent~$i$ does not \emph{overbid} if $b_i \leq \alpha_i\,(p_i - c_i)$, where $p_i$ is the price given as input and $(\alpha_i ,c_i) = \theta_i$ is the true agent $i$'s type.

%
%
%
%
%

Given an allocation $f$, prices $\p$, and $b_i$, we denote with $\widehat{v}_i(f,\p,b_i) = \lambda_{f(i)} \,q_i(p_i,p_{\min}) \, b_i$ the expected (w.r.t.~clicks and purchase) \emph{value} of agent~$i$ according to her declared gain. 
The true expected value that she receives from allocation $f$ is $v_i = \lambda_{f(i)} \,q_i(p_i,p_{\min}) \, \alpha_i \, (p_i - c_i)$, while agent~$i$'s expected \emph{utility} is $u_i = v_i - \pi_i$ since the environment is quasi-linear.\footnote{The dependency of $v_i,u_i,\pi_i$ on the arguments $f,\p,b_i$ is omitted to avoid cumbersome notation.}
The \emph{social welfare} of an allocation with respect to the declared gains is $\widehat{\SW}(f,\p,\bp) = \sum_i \widehat{v}_i(f,\p,b_i)$, where $\bp = \{b_i\}_{i \in N}$. The true social welfare is $\SW = \sum_i v_i$. The \emph{revenue} is instead $\Rev=\sum_i \pi_i$.

Next, we informally introduce notable properties of mechanisms; see~\cite{RePEc:oxp:obooks:9780195102680} for  formal definitions. A mechanism, both direct- and indirect-revelation, is \emph{individually rational}, if for every agent $i$, the assigned payment $\pi_i$ is non-larger than her value $\hat{v}_i(f,\p,b_i)$ according the declared gain. Furthermore, a mechanism is \emph{weakly budget-balanced} if the sum of payments is always non-negative.
A direct-revelation mechanism is \emph{truthful} if for every agent $i$ it is a dominant strategy to report the true type $\theta_i = (\alpha_i, c_i)$ to the mechanism, \emph{i.e.}, the utility that agent~$i$ achieves by reporting $\theta_i$ is at least as large as with every alternative input, regardless of other agents' actions. 
For indirect-revelation mechanisms, we say that a set of inputs is in \emph{equilibrium} according to~\citet{nash1951non} if no agent may increase her utility by submitting a different bid, whenever the inputs of other agents remain unchanged.

\section{Mechanisms}

In this section, we introduce our direct-revelation mechanism and two indirect-revelation mechanisms.

\subsection{Direct-revelation Mechanism}
%
%
We let \dvcg~be the direct-revelation mechanism defined as follows.
%
Given the agent $i$'s input $ \theta'_i = ( \alpha'_i, c'_i) \in \Theta$, the mechanism defines the declared gain $b_i =  \alpha'_i \, (p_i - c'_i)$ for every price $p_i$.
%
Then, the mechanism computes an assignment $f^*$ and prices $\p^*$ that maximize the social welfare with respect to the declared gains; formally, 
\[
\widehat{\SW}(f^*, \p^*, \bp) = \max_{f,\p} \, \widehat{\SW}(f, \p, \bp).
\]
%
Finally, the mechanism assigns to each advertiser $i$ in the allocation (\emph{i.e.}, such that $f(i) \in M$) the VCG payment
\begin{align*}
\pi_i &=  \max_{f,\p} \sum_{\begin{subarray}{c}j \neq i\colon f(j) \in M\end{subarray}} \Big(\widehat{v}_j(f, \p, b_j) - \widehat{v}_j(f^*, \p^*, b_j)\Big)\\
& = \widehat{v}_i(f^*, \p^*, b_i) - \Delta_i,
\end{align*}
where $$\Delta_i = \widehat{\SW}(f^*, \p^*, \bp) - \max_{\begin{subarray}{c}f,\p : f(i) \notin M\end{subarray}} \widehat{\SW}(f, \p, \bp) \geq 0.$$
%
%
It is immediate to check that payments are non-negative, and that they are never larger than the value corresponding to the declared gain. Thus, the mechanism is trivially individually-rational and weakly budget-balanced. Moreover, it is not hard to verify that these payments allow the mechanism to be truthful (essentially, this is a VCG mechanism and there is no interdependence among types). Truthfulness also implies that the mechanism maximizes the true social welfare. These observations prove the following theorem.
\begin{theorem}
Mechanism {\dvcg} is truthful, individually rational, weakly budget-balanced, and maximizes $\SW$.
\end{theorem}

\subsection{Indirect-revelation Mechanisms}

Next, we introduce two alternative indirect-revelation mechanisms, namely {\ivcg} and {\igsp}. 
They share the same structure, but they differ in the way they compute the payments. 
They work as follows.
%
%
Agent $i$ inputs $(p_i, b_i)$, where $p_i \in \mathbb{R}_{\geq 0}$ is the price that agent~$i$ wants to be displayed for her ad and $b_i \in \mathbb{R}$ is the expected gain that $i$ declares to achieve from a click on her ad for price $p_i$.
 %
The mechanism computes an assignment $g^*$ that maximizes the social welfare with respect to the submitted prices and gains; formally
\[
\widehat{\SW}(g^*, \p, \bp) = \max_g \, \widehat{\SW}(g, \p, \bp).
\]
 %

Then, mechanism {\ivcg} assigns to each advertiser $i$ such that $g^*(i) \in M$ the following VCG payment
\begin{align*}
 \pi_i & = \max_{g} \sum_{\begin{subarray}{c}j\neq i \colon g(j) \in M\end{subarray}} \Big(\widehat{v}_j(g,\p,b_j) - \widehat{v}_j(g^*,\p, b_j)\Big)\\
 & = \widehat{v}_i(g^*, \p, b_j) - \delta_i,
\end{align*}
where
$$\delta_i = \widehat{\SW}(g^*,\p,\bp) - \max_{\begin{subarray}{c}g: g(i) \notin M\end{subarray}} \widehat{\SW}(g,\p,\bp) \geq 0.$$
 %
 
 %
%

 
W.l.o.g., let the optimal allocation $g^*$ be such that only the first $\ell \le m$ slots are assigned and no slot $j > \ell$ is assigned.
Then, mechanism {\igsp} assigns to each advertiser $i$ such that $g^*(i) \in M$ and $g^*(i) < \ell$ (\emph{i.e.}, $i$ is assigned to a slot different from $\ell$) the following payments: 
\begin{equation}
	\label{eq:gsp_pay}
	\varpi_i = \lambda_{g^*(i)} q_j(p_j, p_{\min}) b_j,
\end{equation}
 where $j$ is such that $g^*(j)=g^*(i)+1$.
When $g^*(i) = \ell$, there are two possible payments. If all the advertisers $j$ that are \emph{not} assigned to any slot (\emph{i.e.}, such that $g^*(j)=\bot$) have a price $p_j < p_{\min}$, then $\varpi_i = 0$. Otherwise, the payment is
\begin{equation}
	\label{eq:last}
	\varpi_i = \lambda_{g^*(i)} \max_{j : p_j \ge p_{\min} \land g^*(j)=\bot} q_j(p_j, p_{\min}) b_j.
\end{equation}


As for {\dvcg}, it is immediate to check that payments are non-negative, and that they are always less than the value corresponding to the declared gain. Hence, {\ivcg} is individually rational and weakly budget-balanced. Moreover, one may hope that the inputs that agents select at any equilibrium are such that the allocation selected by the mechanism maximizes social welfare. Unfortunately, we will show in the next sections that this is \emph{not} the case.

The payments of {\igsp} are at least zero, and, thus, the mechanism is weakly budged-balanced. Let us also observe that, given an agent $i$, for every other agent $j$ such that either $g^*(j)>g^*(i)$ or $g^*(j)=\bot$ $\land$ $p_j \ge p_{\min}$, we have that $q_j(p_j, p_{\min}) b_j \leq q_i(p_i, p_{\min}) b_i$. 
Otherwise, the allocation $g$ obtained from $g^*$ by fixing $g(j)=g^*(i)$, $g(i)=g^*(j)$, and $g(k)=g^*(k)$ for all $ k \notin \{i,j\}$ would achieve a larger social welfare (according to declared gains). Hence, we have that $\varpi_i \leq \widehat{v}_i(g^*, \p, b_i)$, and, thus, the mechanism is individually rational. We remark that, for this property to hold, it is fundamental that in Equation \ref{eq:last} we only consider the unassigned agents $j$ who have a declared price $p_j\ge p_{\min}$. Indeed, an agent $j$ with $p_j < p_{\min}$ may have a large $q_j(p_j, p_j)b_j$ value, so that, if the $j$-th ad is displayed, the minimum price changes from $p_{\min}$ to $p_j$,  $q_j(p_j, p_j) b_j > q_i(p_i, p_{\min}) b_i$, and $\varpi_i> \widehat{v}_i(g^*, \p, b_i)$, where $i$ is the agent assigned to slot $\ell$.
Nevertheless, this agent may not be chosen by the allocation $g^*$ because of the negative externalities that its low price would put on other agents (by lowering their value and the social welfare). 
As a result, an optimal allocation may not assign all the slots.
We finally observe that, as for {\ivcg}, also {\igsp} may fail to optimize the true social welfare. The following sections will bound the extent of this failure.

\section{Computational Complexity}
In general, externalities make hard the problem of computing an allocation maximizing the social welfare. In this section, we prove that in our setting the problem of allocating advertisers to slots can be solved in polynomial time by both direct- and the indirect-revelation mechanisms.

Let us start with the problem of computing the allocation $g^*$ in the indirect-revelation mechanisms. The following theorem shows that $g^*$ can be computed efficiently.
\begin{theorem}
\label{thm:comp_indir}
 There exists an algorithm that computes the allocation $g^*$ in time $O(n^2 \log n)$.
\end{theorem}
\begin{proof}
 Let $\bp$ and $\p$ be the set of gains and prices submitted by agents.
 First observe that, given a minimum displayed price $p_{\min}$, the allocation that maximizes the social welfare (with respect to gains and prices in input), can be trivially computed by sorting agents in $\{i \colon p_i \geq p_{\min}\}$ in order of $q_i(p_i, p_{\min}) b_i$ and assigning slot $1$ to the agent that maximizes this quantity, slot $2$ to the second such agent, and so on. Note that this operation requires $O(n \log n)$ steps.
 
 However, in order to provide the allocation $g^*$, we also need to decide which is the best value for $p_{\min}$. However, since $p_{\min}$ must belong to $\p$, it is sufficient to compute the best allocation by using as minimum displayed price each of the at most $n$ different prices in $\p$, and choosing the allocation that optimizes the social welfare.
\end{proof}

Computing $g^*$ is an easier problem than the one faced by the direct-revelation mechanism, since, for the former, prices are given and we optimize only over the allocation function, while, for the latter, optimization occurs both on the allocation function and prices.
Nevertheless, the following theorem shows that $f^*$ and $\p^*$ can also be computed efficiently, as long as the set $P$ of allowed prices is discrete and finite.
\begin{theorem}
\label{thm:comp_direct}
 There is an algorithm that computes the allocation $f^*$ and prices $\p^*$ in time $O(n^2|P|(|P|+\log n))$.
\end{theorem}
\begin{proof}
Let $b_i(p) = \alpha'_i(p-c'_i)$ be the expected gain of agent $i$ according to her input when ad $i$ is displayed with price $p$, where $(\alpha'_i, c'_i)$ is the input of agent $i$. For each agent $i$ and every price $\hat{p} \in P$ we compute $p^*_i(\hat{p})$ as follows: if $\max_{\begin{subarray}{c}p \in P : p \geq \hat{p}\end{subarray}} q_i(p, \hat{p}) b_i(p) > 0$, 
then $$p^*_i(\hat{p}) = \arg \max_{\begin{subarray}{c}p \in P:p \geq \hat{p}\end{subarray}} q_i(p, \hat{p}) b_i(p),$$ otherwise we set $p^*_i(\hat{p}) = \bot$.
Roughly speaking, $p^*_i(\hat{p})$ is the best price (according to her input) for agent $i$ when the minimum displayed price is $\hat{p}$ and the $i$-th ad is displayed (and thus $i$'s price is at least $\hat{p}$). Clearly, if there is no price larger than or equal to $\hat{p}$ guaranteeing to agent $i$ a positive utility, then she prefers to be not displayed. For this reason, in the latter case, we do not assign any value to  $p^*_i(\hat{p})$.
Notice that $p^*_i(\hat{p})$ can be computed by evaluating the function for every $p \in P$ with $p \geq \hat{p}$, requiring at most $O(|P|)$ operations.

Then, if the minimum displayed price $p_{\min}$ was given, along with the agent to which it is assigned, then we simply choose price $p^*_i(p_{\min})$ for each remaining agent $i$ (this can be done in $O(nP)$ steps), prune out agents for which $p^*_i(p_{\min}) = \bot$, and finally compute the corresponding optimal assignment by sorting the remaining agents in order of $b_i(p^*_i(p_{\min}))$, as shown in Theorem~\ref{thm:comp_indir} (in $O(n \log n)$ steps).

Unfortunately, selecting $p_{\min}$ is much harder than in the indirect case: not only the value of $p_{\min}$ can assume every value in $P$ (and not just one among at most $n$ alternatives), but we also need to decide which agent should display this price. For this reason, we need to go through every price $p \in P$ and every agent $i$ and compute the best solution that would be achieved when $i$ is the agent displaying the minimum price $p$. Since for each of the $nP$ possible choices, computing the best solution requires time $O(nP+n\log n)$, we achieve the desired running time.
\end{proof}

Observe that the dependence on $|P|$ in Theorem~\ref{thm:comp_direct} is in some way necessary as long as we would like to keep quality function as general as possible. It is not hard to see that we can avoid to check all prices by doing opportune restriction on the quality functions.

We finally highlight that the discretization of the set of prices does not affect the property of the mechanism. In particular, truthfulness continues to hold, since the mechanism is maximal-in-the-range.


\section{Performance of the Indirect Mechanisms}

For the sake of presentation, we provide the informal definitions of PoS and PoA for social welfare and revenue; formal definitions can be found in~\cite{10.5555/1296179}. 
\begin{itemize}
\item PoS for the social welfare is the minimum---w.r.t.~all the Nash equilibria---ratio between the maximum achievable social welfare
and the social welfare of an allocation achievable in a Nash equilibrium of an indirect-revelation mechanism \ivcg{} or \igsp{}.
\item PoA for the social welfare is the maximum---w.r.t.~all the Nash equilibria---ratio between the maximum achievable social welfare
and the social welfare of an allocation achievable in a Nash equilibrium of an indirect-revelation mechanism \ivcg{} or \igsp{}.
\item PoS for the revenue is the minimum---w.r.t.~all the Nash equilibria---ratio between the maximum revenue achievable by an individually-rational mechanism
and the revenue achievable in a Nash equilibrium  of an indirect-revelation mechanism \ivcg{} or \igsp{}.
\item PoA for the revenue is the maximum---w.r.t.~all the Nash equilibria---ratio between the maximum revenue achievable by an individually-rational mechanism
and the revenue achievable in a Nash equilibrium  of an indirect-revelation mechanism \ivcg{} or \igsp{}.
\end{itemize}
\begin{table}[b]
\begin{center}
\resizebox{\columnwidth}{!}{
\renewcommand{\arraystretch}{1.3}
\begin{tabular}{p{0.06\textwidth}||p{0.06\textwidth}|p{0.06\textwidth}|p{0.06\textwidth}||p{0.06\textwidth}|p{0.06\textwidth}|p{0.06\textwidth}}
			&	\multicolumn{3}{|c||}{$1$ slot}								&	\multicolumn{3}{|c}{$m\geq 2$ slots}					\\	\cline{2-7}
			&	\multicolumn{2}{|c|}{Social Welfare}	&	Revenue				&	\multicolumn{2}{|c|}{Social Welfare}	&	Revenue		\\	\cline{2-7}
			&	PoS			&		PoA		&		PoS				&	PoS			&		PoA		&	PoS			\\	\hline \hline
\ivcg			&	$1$			&		$1$ 		&		$1$ $(\spadesuit)$	&	$1$			&		$m$		&	$\infty$		\\	\hline
\igsp			&	$1$			&		$1$		&		$\infty$			& 	$\geq 2$		&		$\geq m$	&	$\infty$		\\	\hline
\end{tabular}
}
\caption{Lower and upper bounds over PoS and PoA when agents do not overbid. $\spadesuit$: PoS here is taken w.r.t.~the mechanism \dvcg{} maximizing the social welfare (thus not necessarily maximizing the revenue).}
\label{table:inefficiencybounds}
\end{center}
\end{table}
Table~\ref{table:inefficiencybounds} summarizes the lower and upper bounds over the mechanisms' inefficiency when agents do not overbid;
the results when agents overbid are omitted since the inefficiency can be arbitrary even with a single slot.
Interestingly, while \ivcg{} performs as well as \dvcg with a single slot as \ivcg{} and \dvcg{} are equivalent in this case since there is no externality; with more than 2 slots the inefficiency can be large both for social welfare and revenue even in the basic case in which slots are indistinguishable and $\lambda = 1$.
In particular, in our proofs of the upper-bound results, we use a special class of quality functions that we denote as \emph{only-min} functions, which assign a value 0 to the quality of an agent when her price is not the minimum among those displayed, and we prove that in many cases no worse instance is possible.
With multiple slots, the positive result is that, with \ivcg{}, the optimal allocation is always achievable by some Nash equilibrium (\emph{i.e.}, $\text{PoS }=1$).
Nevertheless, there are auction instances in which some Nash equilibria lead to allocations whose social welfare is $1/m$ of the optimal allocation (\emph{i.e.}, $\text{PoA}=m$) or in which all the Nash equilibria lead to a revenue of zero whereas the direct-revelation mechanism \dvcg{} provides a strictly positive revenue (\emph{i.e.}, $\text{PoS}=\infty$).
\igsp{} performs even worse than \ivcg{}, both with a single and multiple slots. 
%
%
%

%
In the following, we formally provide the results on the lower and upper bounds over the mechanisms' inefficiency.

\subsection{Price of Stability for the Social Welfare}

Initially, we provide our main positive result in terms of indirect-revelation mechanisms inefficiency.
\begin{theorem}\label{thm:indirect}
The PoS for the social welfare of \ivcg{} is $1$.
\end{theorem}
\begin{proof}
Suppose that each agent $i$ reports the pair $(\tilde{p}_i, \tilde{b}_i)$ defined as follows: if the mechanism {\dvcg} displays the ad~$i$ when run on truthful bids, then $\tilde{p}_i$ is the corresponding price, and $\tilde{b}_i = \alpha_i (\tilde{p}_i - c_i)$, \emph{i.e.}, the true gain associated to this price; otherwise $\tilde{p}_i = \tilde{b}_i = 0$. It is immediate to check that with these bids the allocation returned by {\ivcg} is exactly the same as the one returned by {\dvcg}, and, thus, it maximizes social welfare.

Unfortunately, we cannot conclude that inputs $(\tilde{p}_i, \tilde{b}_i)$ are in equilibrium directly from the truthfulness of {\dvcg}.
Indeed, the payments assigned by the indirect mechanism are different from the ones assigned by the direct mechanism. Moreover, in the former the agent may lie both about the price and about the expected gain, while in the latter an agent may essentially lie only on the expected gain. Still, in the following we prove that inputs $(\tilde{p}_i, \tilde{b}_i)$ are in equilibrium, and, thus, the theorem follows.

In particular, let $\tilde{\p} = (\tilde{p}_1, \ldots, \tilde{p}_n)$ and $\tilde{\bp} = (\tilde{b}_1, \ldots, \tilde{b}_n)$. We prove that the utility $\tilde{u}_i$ of agent $i$ when the mechanism {\ivcg} is run on $\tilde{\p}$ and $\tilde{\bp}$ is at least the utility $u_i$ that she achieves if the mechanism would be run on input $\p=(p_i, \tilde{\p}_{-i})$ and $\bp=(b_i,\tilde{\bp}_{-i})$, for every $i$, $p_i$, and $b_i$.
Indeed if $i$ is allocated by the mechanism {\ivcg} when run on input $\tilde{\p}$ and $\tilde{\bp}$, then, since, by definition of $\tilde{b}_i$, $v_i = \widehat{v}_i(f^*,\tilde{\p},\tilde{b}_i)$,
\begin{align*}
 \tilde{u}_i & = v_i - \pi_i = \widehat{v}_i(f^*,\tilde{\p},\tilde{b}_i) - \pi_i\\
 & = \widehat{SW}(f^*,\tilde{\p},\tilde{\bp}) - \max_{\begin{subarray}{c}g: g(i) \notin M\end{subarray}}\widehat{SW}(g,\tilde{\p},\tilde{\bp}) \geq 0,
\end{align*}
where $f^*$ is the allocation returned by {\dvcg} on truthful bids. If $i$ is instead, unallocated then
$$
 \tilde{u}_i = 0 = \widehat{SW}(f^*,\tilde{\p},\tilde{\bp}) - \max_{\begin{subarray}{c}g: g(i) \notin M\end{subarray}}\widehat{SW}(g,\tilde{\p},\tilde{\bp}).
$$

Thus, if the agent $i$ is unallocated by the mechanism {\ivcg} when run on input $\p$ and $\bp$, then the equilibrium condition is trivially satisfied.
Otherwise, let $\check b_i=\alpha_i (p_i-c_i)$ and $\check \bp=(\check b_i, \bp_{-i})$. We have:
\begin{align*}
u_i & = v_i - \pi_i  = \widehat{v}_i(g^*,\p,\check b_i) - \widehat{v}_i(g^*,\p,b_i)\\
& \phantom{= v_i - \pi_i  =} + \widehat{SW}(g^*,\p,\bp) - \max_{\begin{subarray}{c}g:g(i) \notin M\end{subarray}}\widehat{SW}(g,\p,\bp)\\
& = \widehat{SW}(g^*,\p,\check{\bp}) - \max_{\begin{subarray}{c}g:g(i) \notin M\end{subarray}}\widehat{SW}(g,\tilde{\p},\tilde{\bp}),
\end{align*}
where the last equality follows since $p_j = \tilde{p}_j$ and $b_j = \tilde{b}_j$ for every agent $j \neq i$.

Since $\widehat{SW}(f^*,\tilde{\p},\tilde{\bp}) \geq \widehat{SW}(g^*,\p,\check{\bp})$, because $f^*$ and $\tp$ are the allocation and the prices that maximize the social welfare, we have that $\tilde{u}_i \geq u_i$, as desired.
\end{proof}

The proof of the theorem above shows that, with VCG payments, there is always a Nash equilibrium in which every agent~$i$ bids the truthful gain $b_i$ and the price that \dvcg{} would use. Such a strategy profile leads to the same allocation of \dvcg{}, thus guaranteeing a PoS for the social welfare of 1, but, as we discuss in the following sections, the revenue of the two mechanisms can be different. The same result does not hold in the case of GSP payments, thus leading to a larger PoS for the social welfare.
\begin{restatable}{theorem}{thmfive}
The PoS for the social welfare of \igsp{} is at least $ 2$ even if agents do not overbid.
\end{restatable}

\subsection{Price of Anarchy for the Social Welfare}
We initially focus on the basic case with a single slot, showing that in this case \ivcg{} and \igsp{} are efficient.\footnote{All the missing proofs are in the Extended Version.}

\begin{restatable}{theorem}{thmsix}
\label{theorem:POA-VCG-GSP-single-slot}
The PoA for the social welfare of \ivcg{} and \igsp{} is $1$ if $m = 1$ when agents do not overbid.
\end{restatable}


Then, we study the case with multiple slots providing a lower bound on PoA.

\begin{restatable}{theorem}{thmseven}
\label{theorem:POA-VCG-GSP-atleast-m}
The PoA for the social welfare of \ivcg{} and \igsp{} is at least $m$ if $m \geq 2$ when agents do not overbid.
\end{restatable}

In the specific case of \ivcg{}, we show that a PoA larger than $m$ is not possible, and therefore there are no instances worse than those used in the proof of Theorem~\ref{theorem:POA-VCG-GSP-atleast-m}. 
Most interestingly, this result holds even when $q_i$ is not monotonically decreasing in $p_i$.
\begin{restatable}{theorem}{thmeight}
The PoA for the social welfare of \ivcg{} is at most $m$ if $m \geq 2$ when agents do not overbid.
\end{restatable}

Finally, we show that when agents overbid, the inefficiency can be arbitrarily large.
\begin{restatable}{theorem}{thmnine}
The PoA for the social welfare of \ivcg{} and \igsp{} is $\infty$ even if $m=1$ when agents can overbid.
\end{restatable}

\subsection{Price of Stability for the Revenue}

Initially, we provide our main result, showing that \ivcg{} and \igsp{} can be arbitrarily inefficient even with 2 slots.
\begin{restatable}{theorem}{thmten}
The PoS for the revenue of \ivcg{} and \igsp{} is $\infty$ even if $m=2$.
\end{restatable}

In the specific case of \ivcg{} and $m=1$, we have a positive result for PoS (PoA is trivially $\infty$ as it is $\infty$ even in second-price single-item auctions).
\begin{restatable}{theorem}{thmeleven}
The PoS for revenue of \ivcg{}  with respect to the mechanism \dvcg{} is $1$ if $m = 1$.
\end{restatable}

 
 Instead, the above positive result does not hold with \igsp{}, as stated below.
\begin{restatable}{theorem}{thmtwelve}
The PoS for the revenue of \igsp{} is $\infty$ even if $m = 1$ when agents do not overbid.
\end{restatable}
In the proof of this theorem we strongly rely upon the definition of GSP payments described above, which restricts payments to depend only on agents submitting a price at least $p_{\min}$. This payment format turns out to be necessary in order to guarantee individual rationality. We leave open the problem of understanding if a better Price of Stability for the revenue of {\igsp} would be possible by considering alternative non-individually rational GSP payments.

%
%
%
%

%

\section{A Better PoS for the Revenue with Indirect-revelation Mechanisms}

As discussed in the previous section, indirect-revelation mechanisms present major weaknesses in terms of efficiency. 
A natural question is whether we can design indirect-revelation mechanisms with a better efficiency when agents can choose their price.
In particular, we focus on \ivcg{}, as it always guarantees $\text{PoS} = 1$ for the social welfare, and we show that a simple modification of the mechanism leads to $\text{PoS} = 1$ for the revenue when some assumptions hold.
We call this new mechanism \ivcgstar.
The rationale is to ask agents for more information.
More precisely, the input provided by every agent is a triple composed of $(b_i, p_i, p_i^*)$ where $(b_i,p_i)$ is the input to \ivcg{} and $p^*_i$ is the price that advertiser~$i$ would choose when her ad is the only displayed ad. 
The property that $\text{PoS} = 1$ is guaranteed when function $q_i(p_i,p_i)$ is differentiable in $p_i$ and non-zero in $p_i^*$.
Mechanism \ivcgstar{} is defined as follows:
\begin{enumerate}
\item every agent $i$ submits a bid $(b_i,p_i,p^*_i)$, where $b_i,p_i$, and $p_i^*$ are defined as above;
\item the mechanism infers the values of $c_i$ and $\alpha_i$  for every agent~$i$ as follows: $\hat{c}_i=q(p^*_i,p^*_i)\,/\,\frac{dq(p_i,p_i)}{dp_i}\big\vert_{p_i=p_i^*}+p^*_i$, and $\hat{\alpha}_i=\frac{b_i}{p_i-\hat{c}_i}$ if $p_i\neq \hat c_i$ and $\hat \alpha_i=0$ otherwise;
\item the mechanism computes an auxiliary allocation, say $\bar{f}$, by using the allocation function of \ivcg{} when the input is $(b_i,p_i)$ for every agent~$i$; the corresponding social welfare (evaluated with the declared gain $b_i$) is $\widehat{\SW}$;
%
\item for every agent $i$, the mechanism computes an auxiliary allocation, say $\bar{f}^{-i}$, by using the allocation function of \dvcg{} when the values inferred above for $\{\hat{\alpha}_h\}_{h \in N}$ and $\{\hat{c}_h\}_{h \in N}$ are provided in input and agent~$i$ is removed from the optimization problem.
For every maximization, we denote with $\overline{\SW}^{-i}$ the corresponding social welfare evaluated with the inferred values $\{\hat \alpha_h\}_{h \in N}$ and $\{\hat c_h\}_{h \in N}$.
Notice that, as it happens with \dvcg, the prices in output to these maximizations can be different from those agents provide in input;
\item if $\widehat{\SW} \geq \max_i \overline{\SW}^{-i}$, then the mechanism chooses allocation $\bar{f}$ and charges every agent~$i$ of a payment $\pi_i = \overline{\SW}^{-i} - (\widehat{\SW} - \lambda_{\bar{f}(i)}\, q_i(p_i,p_{\min}) \, b_i)$, else no ad is allocated and every agent is charged a payment of zero.
\end{enumerate}
Basically, mechanism \ivcgstar{} exploits the additional information asked to the agents to infer their types and then uses this information to compute the same payments that \dvcg{} would charge.
Step~5 is necessary to guarantee individual rationality. More precisely, since the allocation $\bar{f}$ is computed as the indirect mechanism does (without optimizing over prices), while the payments $\{\pi_i\}_{i \in N}$ are computed as the direct mechanism does (optimizing over prices), individual rationality may not be satisfied. 
We solve this problem setting the payments to $0$ (and allocating no ads) when the payments $\{\pi_i\}_{i \in N}$ are too large.
As a side effect, we have that if the submitted prices are different from the optimal one, it is possible that the mechanism does not assign any slot. Thus, the PoA for the social welfare and revenue can be unbounded. 
\begin{restatable}{theorem}{ivcgstarthm}
\label{theorem:ivcgstar}
	Mechanism \ivcgstar{} is individually rational and weakly budget-balanced. Moreover, the PoS for the revenue of \ivcgstar{} is 1.
\end{restatable}

We recall that the algorithm we provide to find the best allocation with \ivcg{} works when the values that $p_i$ can assume are discrete, and the same holds with \ivcgstar{}.
We also notice that \ivcgstar{} requires that $p^*_i$ is not restricted to a set of discrete values, the mechanism could not infer the exact values of $\alpha_i$ and $c_i$ otherwise.
However, requiring price $p_i$ to belong to a finite, discrete set of values and price $p_i^*$ to belong to $\mathbb{R}_{\geq 0}$ does not modify the properties of the mechanism since $p^*_i$ is not used in the allocation algorithm.

%
%

\section{Conclusions and Future Work}

In this paper, we investigate how displaying prices together with ads affects the users' behavior and the properties of auction mechanisms. 
Since the goods sold by the agents are similar, a high competition among the agents arises from the price comparison.
%
Technically speaking, the prices introduce externalities as the probability with which a user clicks on an ad depends on the price of that ad and on the prices of the other displayed ads. 
Interestingly, the social welfare can be maximized when a direct-revelation mechanism jointly optimizes over the ad allocation and the prices, and we show that this can be done in polynomial time when the prices can assume a finite set of values.
However, in practice, it is unlikely that advertisers would allow the mechanism to choose prices on their behalf and, in commonly-adopted mechanisms, ads allocation and price optimization are decoupled, so that the advertisers optimize prices and bids, while the mechanism does so for the allocation, once prices and bids are given.
We show that this decoupling makes standard mechanisms with VCG and GSP payments highly inefficient in terms of PoA and PoS for social welfare and revenue. 
%
%
Finally, we investigate whether we can reduce the inefficiency of mechanisms in which the advertisers optimize prices and bids.
We show that we can obtain PoS of 1 for the revenue by a simple modification of the mechanisms.
In particular, we ask the advertisers for an additional price that the mechanism exploits to infer the values of some advertisers' parameters.
Such a modification can be easily implemented in practice without agents revealing their private, sensitive information.

Many research directions can be explored in future.
Probably, the most interesting concerns how the bidding strategies commonly adopted for standard ad auctions without prices can be extended to our case.
In particular, the crucial question is whether, as in the case of the standard GSP without prices, there are bidding strategies converging to notable Nash equilibria.
Other interesting questions concern the analysis of PoA and PoS and the design of allocation algorithms when the quality functions satisfy specific properties, such as, \emph{e.g.}, smoothness.


\section*{Acknowledgments}	
This work has been partially supported by the Italian MIUR PRIN 2017 Project ALGADIMAR ``Algorithms, Games, and Digital Market''.

\bibliography{biblio}

\clearpage

\appendix

\section{Omitted Proofs}

\thmfive*

\begin{proof}
We will next show a setting for which it occurs that, with GSP payments, all equilibrium bids make the mechanism to allocate agents with very low prices, implying a corresponding low social welfare, where the optimal allocation only allocates agents with high prices.

 \textsc{Setting}.
For $\varepsilon > 0$, consider the following setting:
\begin{itemize}
\item $n = 3$, $m=2$, $P= [\underline{p}, \overline{p}]$
with $\underline{p} > \varepsilon$ and
$\overline{p} = \frac{3}{2}(\underline{p}-\varepsilon)$;
\item for every agent $i$ we have
$$
q_i(p_i,p_{\min}) = 
\begin{cases}
1	          & 	\text{if } p_i = p_{\min},\\
0	          & 	\text{otherwise};
\end{cases}$$
\item $c_1 = 0$, $c_2 = c_3 = \underline{p}-\varepsilon$;
\item for every agent~$i \in N$, $\alpha_i = 1$;
\item  for every slot $j \in M$, $\lambda_j = 1$;
\item ties are broken in favour of agent $1$.
\end{itemize}

Observe that for every $\underline{p} \leq p_1 < p_2 \leq \overline{p}$, it holds that
\begin{equation}
\label{eq:1wins}
 p_1 - c_1 \geq \underline{p}-c_1 = \underline{p} > \underline{p}-\varepsilon \geq 2(p_2-c_2),
\end{equation}
i.e. the social welfare achieved by displaying agents $1$ at price $p$ is larger than what we achieve by displaying agents $2$ and $3$ at price $p_2$. In other words, the mechanism always chooses the price submitted by agent 1 as minimum price.

\textsc{Best Social Welfare with {\igsp}.}
Let $(p_1, p_2, p_3)$ and $(b_1, b_2, b_3)$ be the price and gains given in input to the mechanism. Suppose that they form a Nash equilibrium. We will next provide a characterization of these values.

First observe that, in equilibrium, it must be the case that $b_1$ is large enough to allow ad 1 to be displayed. Indeed, if this is not the case, then agent 1 would have an incentive to submit the true gain, and thus, by \eqref{eq:1wins} and the no overbidding assumption, to be displayed and to achieve a strictly positive utility.

Suppose first that $p_1 = p > \underline{p}$. 
We next show that in this case there is at least one agent $i \in \{2, 3\}$ such that $p_i = p$ and $b_i = p - c_i$ (a larger declared gain would not be possible because of the no overbidding assumption). Suppose indeed that this is not the case. If both agents have $p_i \neq p$, then they must have $0$ utility (if $p_i > p$, their value must be $0$ because of the quality function, and if $p_i < p$ they are not displayed since otherwise they will zeroth the value of agent $1$). However, if one of these agents submits price $p$ and the corresponding true gain, she would be displayed and achieve strictly positive utility, regardless of $b_1$ (if $b_1 \geq p-c_2$, the payment assigned to the deviating agent is 0, and if $b_i < p-c_2$, then the payment will be less than the value for being displayed at that price.)
Suppose then that there is agent $i \in \{2, 3\}$ with $p_i = p$ but $b_i < p - c_i$ and agent $j = 5-i$ with either $p_j \neq p$ or $b_j < p - c_j$. Note that the ad of one of these agents, say, wlog, $j$, is not displayed. Then $j$ has an incentive to submit price $p$ and gain $p - c_j$, since it would assure that her ad will displayed and she receives strictly positive utility.

Next we prove that $b_1 \geq p-c_2$, and thus agent 1 is assigned the first slot (because of the tie-breaking rule). Suppose instead that $b_1 < p - c_2$, and let $i \in \{2, 3\}$ be the one agent with $p_i = p$ and $b_i = p-c_i$. Since, as observed above, $b_1$ must be large enough to have that ad 1 is displayed, this means that either $p_j \neq p$ or $b_j \leq b_1$, with $j = 5-i$. However, as showed above, $j$ has an incentive to deviate by submitting price $p$ and gain $b_j \in (b_i, b_1)$.

Hence, if $p_1 = p > \underline{p}$, then agent 1 will be displayed in the first slot and will be assigned a payment $p-c_2$. Hence, her utility is $c_2 = \underline{p}-\varepsilon$. We will next show that agent $1$ has then an incentive to deviate from this equilibrium. Specifically, let $i \in \{2, 3\}$ be the agent submitting price $p$ and gain $p-c_i$. We distinguish two cases based on $p_j$ and $b_j$, where $j = 5-i$: if $p_j > \underline{p}$ or $b_j < \underline{p} - c_2 = \varepsilon$, then agent 1 has an incentive to submit price $\underline{p}$ and the corresponding true gain, being allocated in the first slot and being assigned a payment of at most $b_j$, resulting in an utility $\underline{p} - b_j > \underline{p} - \varepsilon$; if $p_j = \underline{p}$ and $b_j = \varepsilon$ ($b_j$ cannot be larger because of the no overbidding assumption), then agent 1 has an incentive to submit price $\underline{p}$ and gain $b_1 < \varepsilon$, being allocated in the second slot and receiving a null payment, resulting in utility $\underline{p} > \underline{p} - \varepsilon$.

We can then conclude that in an equilibrium $p_1 = \underline{p}$ and ad $1$ must be displayed, that implies that every equilibrium cannot have social welfare larger than $(\underline{p} - c_1) + (\underline{p} - c_2) = \underline{p}+\varepsilon$.

\textsc{Social Welfare of the Optimal Allocation}.
The optimal allocation will display agents $1$ and $2$ at price $\overline{p}$.
Hence, the optimal social welfare is $(\overline{p}-c_1)+(\overline{p}-c_2)=\frac{3}{2}(\underline{p}-\varepsilon) + \frac{3}{2}(\underline{p}-\varepsilon) - \underline{p} + \varepsilon = 
2(\underline{p}-\varepsilon)$.
Hence, the price of Stability is $\frac{2(\underline{p}-\varepsilon)}{\underline{p}+\varepsilon}$ that goes to $2$ as $\varepsilon$ goes to 0.
\end{proof}

\thmsix*

\begin{proof}
 When a single slot is available, the value of displayed agent $i$ is $\lambda_1 q_i(p_i, p_i) \alpha_i(p_i - c_i)$, where $p_i$ is the corresponding displayed price. That is, this value does not depend on the prices submitted by other agents. Let $\tilde{v}_i = \max_p \lambda_1 q_i(p, p) \alpha_i(p - c_i)$ and set $\tilde{p}_i$ to any price $p$ such that $\lambda_1 q_i(p, p) \alpha_i(p - c_i) = \tilde{v}\}$. Finally, sort agents in order of $\tilde{v}_i$, so that $\tilde{v}_1 \geq \tilde{v}_2 \geq \cdots \geq \tilde{v}_n$. Note that if these values are all equals, then, regardless of the displayed agent, the mechanism always maximizes the social welfare. Suppose instead that there are at least two different values. Let $k$ be the first index such that $\tilde{v}_k > \tilde{v}_{k+1}$. Then we claim that in any equilibrium one agent $i \leq k$ must be displayed, otherwise she has the incentive to submit price $\tilde{p}_i$ and the corresponding true gain. This indeed causes the mechanism to display ad $i$, that provides this agent with a value $\tilde{v}_i$, and to assign a payment (both in case of VCG and GSP payments) that is at most $\tilde{v}_{k+1}$ (because of the non-overbidding assumption), resulting in this way in a positive utility.
\end{proof}

\thmseven*

\begin{proof}
The proof is based on the following setting, in which the ratio between the social welfare of the optimal allocation and the social welfare of the allocation achievable in the worst Nash equilibria in \ivcg{} and \igsp{} is exactly~$m$.

\textsc{Setting}.
Consider the following setting:
\begin{itemize}
\item $n=m+1$;
\item for every agent~$i \in N$, 
\[q_i(p_i,p_{\min}) = 
\begin{cases}
1	& 	\text{if } p_i \leq \overline{p} \text{ and } p_i = p_{\min}	\\
0	& 	\text{otherwise}
\end{cases};\]
\item for every agent~$i \in N$, $c_i = 0$;
\item for every agent~$i \in N$, $\alpha_i = 1$;
\item  for every slot $j \in M$, $\lambda_j = 1$.
\end{itemize}

\textsc{Social welfare of the optimal allocation}.
One of the allocations maximizing the social welfare is such that $f(i)=i$ for every $i \in N, i \neq m+1$ and $f(m+1) = \bot$, while $p_i = \overline{p}$ for every $i \in N$. 
The optimal social welfare $\SW$ is $m \, \overline{p}$.

\textsc{Social welfare with \ivcg{}}.
Define $\underline{p} = \overline{p} / m$. 
Consider the case in which, for every $i \in N$, it holds $p_i = \underline{p}$ and $b_i = \underline{p}$ and therefore every agent is declaring her true gain. 
For every $i \in N$, we have that $q_i = 1$, $p_i = \underline{p}$, and $u_i = 0$, as payment $\pi_i$ equals the expected value $v_i$.
This strategy profile leads to a social welfare $\SW = m \, \underline{p} = m \, \overline{p} / m = \overline{p}$.
In the following, we show that such a strategy profile is a Nash equilibrium of the full-information game in which the payments are VCG-like, thus proving the theorem.
Initially, we analyze possible deviations to values of $b_i$ different from $\underline{p}$ when keeping $p_i= \underline{p}$, showing that no deviation allows agents to strictly increase their utility.
Since $b_i = \underline{p}$ is the true gain of agent~$i$ from being allocated and agents are assumed not to overbid, no agent~$i$ declares a gain larger than $\underline{p}$.
Furthermore, declaring a gain strictly smaller than $\underline{p}$ is a weakly dominated strategy.
Indeed, agent~$i$ with $f(i)\in M$ would be not displayed by declaring a gain smaller than $\underline{p}$, while agent~$i$ with $f(i)= \bot$ keeps not be displayed when declaring a gain less than $\underline{p}$.
Now, we analyze possible deviations to values of $p_i$ different from $\underline{p}$. 
We consider the restricted case in which, in the deviation, agent~$i$ changes both $b_i$ and $p_i$ such that $b_i = p_i$, discussing below that no deviation with $b_i < p_i$ is useful for agent~$i$.
Notice that, by setting $b_i = p_i$, an agent is declaring exactly her gain, and, therefore, any strictly larger declared gain would correspond to overbidding.
For every $p_i > \underline{p}$, we have that $q_i=0$ if ad~$i$ is displayed together other ads, as the other ads have a price strictly smaller than $p_i$.
Thus, either ad~$i$ is displayed alone in the allocation, to guarantee that $p_i$ is the minimum price of this new allocation and therefore that $q_i >0$, or ad~$i$ is not displayed as we can allocate $m$ ads each with a strictly positive value. 
Under the assumption that ties are opportunely broken, ad~$i$ is displayed alone only if her gain (\emph{i.e.}, $p_i$) is strictly larger than the cumulative gain of the other ads (\emph{i.e.}, $m\, \underline{p}$).
This never happens as, by construction, $p_i \leq m\, \underline{p} = \overline{p}$, otherwise (\emph{i.e.}, for $p_i > \overline{p}$) the value of $q_i$ would be $0$, and therefore $b_i\leq m\, \underline{p}$, not allowing ad~$i$ to be displayed.
Notice that the same happens when, in the deviation, agent~$i$ underbids making $b_i < p_i$.
Finally, by arguments similar to those used above, if $p_i < \underline{p}$, ad~$i$ is either displayed alone or not displayed. 
As above, when $b_i \leq p_i< \underline{p}$, agent~$i$ cannot be displayed as her gain cannot be larger than the cumulative gain of the other agents.

\textsc{Social welfare with \igsp{}}.
Consider the case in which, for every $i \in N$, $p_i = \underline{p}$ and $b_i = \underline{p}$. 
For the same arguments used above for \ivcg, such a strategy profile is a Nash equilibrium, thus leading to a social welfare that is $1/m$ of the optimal social welfare.
This concludes the proof.
\end{proof}

\thmeight*

\begin{proof}
For the sake of presentation, we introduce the following notation:
\begin{itemize}
\item we denote the maximum value agent~$i$ can get with $v_i^* = \lambda_1 \,\max_{p_i}\{q_i(p_i,p_i) \, \alpha_i \, (p_i -c_i)\}$;
\item we denote the corresponding optimal price with $p_i^*$;
\item we denote the corresponding true gain with $b_i = \alpha_i\,(p^*_i - c_i)$;  
\item we denote the allocation in a Nash equilibrium with $f^{\text{NE}}$;
\item we denote the declared gain used by agent~$i$ in the Nash equilibrium with $b_i^{\text{NE}}$;
\item we denote the price used by agent~$i$ in the Nash equilibrium with $p^{\text{NE}}_{i}$;
\item we denote the minimum price among those of the displayed ads in the Nash equilibrium with $p^{\text{NE}}_{\min}$;
\item we denote the value agent~$i$ gets in the Nash equilibrium with $v_i^{\text{NE}} = \lambda_{f^{\text{NE}}(i)}\,q_i(p_i^{\text{NE}},p_{\min}^{\text{NE}})\,\alpha_i\,(p_i^{\text{NE}}-c_i)$;
\item we denote the payment of agent~$i$ in the Nash equilibrium with $\pi_i^{\text{NE}}$;
\item we denote the social welfare in a Nash equilibrium with $\SW^{\text{NE}} = \sum_{i \in N}v_i^{\text{NE}}$;
%
%
%
\item we denote the optimal social welfare when agent~$i$ is discarded evaluated by $\{b_i^{\text{NE}}\}_{i \in N}$ with $\widehat{\SW}_{-i}$.
\end{itemize}
Initially, we prove that, for every Nash equilibrium and agent $i \in N$, it holds $v_i^* \leq \SW^{\text{NE}}$.
According to the definition of the VCG payments, the utility of agent~$i$ in a Nash equilibrium can be written as:
\begin{multline*}
\underbrace{\lambda_{f^{\text{NE}}(i)}\, q_i(p_i^{\text{NE}},p_{\min}^{\text{NE}})\, \alpha_i\, (p_i^{\text{NE}}-c_i)}_{v_i^{\text{NE}}} - \\ \underbrace{\widehat{\SW}_{-i}+  \sum_{h \neq i} \lambda_{f^{\text{NE}}(h)}\, q_h(p_h^{\text{NE}},p_{\min}^{\text{NE}})\, b_h^{\text{NE}}}_{\pi_i^{\text{NE}}}= 
\end{multline*}
\begin{multline*}
\Bigg(\lambda_{f^{\text{NE}}(i)}\, q_i(p_i^{\text{NE}},p_{\min}^{\text{NE}})\, \alpha_i\, (p_i^{\text{NE}}-c_i) + \\ \sum_{h \neq i} \lambda_{f^{\text{NE}}(h)}\, q_h(p_h^{\text{NE}},p_{\min}^{\text{NE}})\, b_h^{\text{NE}}\Bigg) -  \Bigg(\widehat{\SW}_{-i}\Bigg).
\end{multline*}
We call $\widehat{\SW}_{-i}$ as negative-utility term and the remaining part as positive-utility term.
%
Since agents are not allowed to overbid, we have that the positive-utility term is $\leq \SW^{\text{NE}}$.
Notice that negative-utility term $\widehat{\SW}_{-i}$ is a constant for every deviation of agent~$i$ from $(b_i^{\text{NE}},p_i^{\text{NE}})$, not depending on $(b_i,p_i)$.
This means that agents aim at maximizing the positive-utility term.
Therefore, there is no pair $(b_i,p_i)$ that agent~$i$ can play leading to a different allocation $f$ and/or minimum price $p_{\min}$ providing a positive-utility term strictly larger than $\lambda_{f^{\text{NE}}(i)}\, q_i(p_i^{\text{NE}},p_{\min}^{\text{NE}})\, \alpha_i\, (p_i^{\text{NE}}-c_i) + \sum_{h \neq i} \lambda_{f^{\text{NE}}(h)}\, q_h(p_h^{\text{NE}},p_{\min}^{\text{NE}})\, b_h^{\text{NE}}$, otherwise agents would not play a Nash equilibrium.
Among all the possible deviations of agent~$i$, we have that the deviation towards $(b_i^*,p_i^*)$ would provide a positive-utility term $\geq v^*_i$. Indeed, it is always possible to allocate ad~$i$ in slot~$1$ and not to display the other ads, thus obtaining exactly $v^*_i$, or allocating other ads in addition to ad~$i$, thus obtaining strictly more than $v^*_i$.
However, the positive-utility term provided when agent~$i$ deviates towards $(b_i^*,p_i^*)$ is not larger than the positive-utility term in the Nash equilibrium, otherwise we would not be in a Nash equilibrium.
Therefore, we have:
\begin{multline*}
v^*_i \leq
\lambda_{f^{\text{NE}}(i)}\, q_i(p_i^{\text{NE}},p_{\min}^{\text{NE}})\, \alpha_i\, (p_i^{\text{NE}}-c_i) + \\ \sum_{h \neq i} \lambda_{f^{\text{NE}}(h)}\, q_h(p_h^{\text{NE}},p_{\min}^{\text{NE}})\, b_h^{\text{NE}} \leq
 \SW^{\text{NE}}.
\end{multline*}

Now, we show that the value of the optimal allocation, say $\text{OPT}$, is $\leq m \, \SW^{\text{NE}}$.
Indeed, we have:
\[
\text{OPT} \leq \sum_{i \in \text{top}(m)} v_i^* \leq m \, \SW^{\text{NE}},
\]
where $\text{top}(m)$ are the ads with the top $m$ values $v_i^*$. 
The first inequality holds as $\text{OPT}$ cannot be larger than the sum of the top $m$ values $v_i^*$.
The second inequality follows from what showed above.
Thus, it follows $\text{OPT} \leq m\, \SW^{\text{NE}}$.
\end{proof}

\thmnine*

\begin{proof}
The proof is based on the following setting.
\textsc{Setting}.
Consider the following setting:
\begin{itemize}
\item $n=2$;
\item $m=1$;
\item $q_1(p_1,p_{\min}) = 
\begin{cases}
1	& 	\text{if } p_1 \leq \overline{p} \text{ and } p_1 = p_{\min};	\\
0\;\;\;\;\;\;\,	& 	\text{otherwise};
\end{cases}$
\item $q_2(p_2,p_{\min}) = 
\begin{cases}
\delta	& 	\text{if } p_2 \leq \overline{p} \text{ and } p_2 = p_{\min};	\\
0\;\;\;\;\;\;\,	& 	\text{otherwise};
\end{cases}$\\
where $0<\delta < 1$;
\item $c_1 = c_2 = 0$;
\item $\alpha_1 = \alpha_2 = 1$;
\item $\lambda_1 = 1$.
\end{itemize}
\textsc{Social welfare of the optimal allocation}.
The optimal allocation is $f(1)=1, f(2)=\bot$ with $p_1=p_2 = \overline{p}$.
\textsc{Social welfare with \ivcg{} and \igsp{}}.
When overbidding is allowed, the following strategy profile is a Nash equilibrium $(b_1 = 0, p_1 = \overline{p}, b_2 = 2\,\overline{p}/\delta, p_2 = \overline{p})$.
Let us notice that indirect-revelation VCG and GSP mechanisms are the same mechanism in this case.
Indeed, agent~$2$ gets a utility of $\delta$, as her payment is $\pi_2 = 0$, and there is no other strategy providing agent~$2$ a strictly larger utility.
Agent~$1$ can get her ad allocated, but, in doing that, she would be charged of a payment $\pi_1 = 2 \,\overline{p}$ strictly larger than her value.
Thus, agent~$1$ will not do it.
As a result, strategy profile $(b_1 = 0, p_1 = \overline{p}, b_2 = 2\, \overline{p}/\delta, p_2 = \overline{p})$ is a Nash equilibrium.
Thus, the Price of Anarchy is $1/\delta$, that is unbounded from above as $\delta \rightarrow 0^+$.
\end{proof}

\thmten*

\begin{proof} 
The proof is based on the following setting, in which \dvcg{} provides a strictly positive revenue, while both \ivcg{} and \igsp{} provide a revenue of zero in every Nash equilibrium.
\textsc{Setting}. Consider the following setting:
\begin{itemize}
\item $n=m=2$;
\item $q_1(p_1,p_{\min}) = 
\begin{cases}
1	& 	\text{if } p_1 \leq \overline{p} \text{ and } p_1 = p_{\min}	\\
0& 	\text{otherwise}
\end{cases}$;
\item $q_2(p_2,p_{\min}) = 
\begin{cases}
1		& 	\text{if } p_2 < \underline{p} \text{ and } p_2 = p_{\min}	\\
\psi(p_2)	& 	\text{if } \underline{p} \leq p_2 \leq \overline{p} \text{ and } p_2 = p_{\min}	\\
0		& 	\text{otherwise}
\end{cases}$ \;\; where $1 \leq \underline{p} <\overline{p}/2$;
\item $c_1 = c_2 = 0$;
\item $\alpha_1 = \alpha_2 = 1$;
\item $\lambda_1 = \lambda_2 = 1$.
\end{itemize}
Furthermore, function $\psi(p_2)$ is defined as follows:
\[
\psi(p_2)= 
\begin{cases}
1				& \text{if } p_2 = \underline{p};		\\
\dfrac{\left(1+\delta\right)\, \left(\overline{p} - \underline{p}\right)}{\left(1 - \frac{1+\delta}{2}\right)\,p_2 + \frac{1+\delta}{2}\,\overline{p}-\underline{p} } -1				& \text{if } \underline{p} < p_2 < \overline{p}; 							\\
\delta			& \text{if }  p_2 = \overline{p};
\end{cases}
\]
where $0 < \delta < \underline{p}/\overline{p}< 1/2$. 
In particular, function $\psi(p_2)$ is an hyperbola such that $\psi(\underline{p})= 1$ and $\psi(\overline{p})= \delta$ and is continuous and monotonically decreasing in $ [\underline{p}, \overline{p}]$. 

\textsc{Revenue with \dvcg}.
We notice that, by definition of $q_1,q_2$, if $p_i = p_{\min}< p_{-i}$, then $q_{-i} = 0$. \footnote{We denote with $-i$ the advertiser $j\neq i.$} 
Therefore, when $p_i < p_{-i}$, no more than one ad with strictly positive value can be allocated.
That is, in this case, the allocation is composed of a single ad.
In particular, the ad~$j$ with the maximum $q_j(p_{j}, p_{j}) \, p_{j}$ is the only ad to be allocated.
However, it can be trivially observed that, when $p_i < p_{-i} \leq \overline{p}$, the mechanism can always increase the value of the allocation by changing the price of the ad that it would not allocate.
In particular, the mechanism can set that price equal to the price of the ad that it would allocate. 
In this way, both ads will be displayed, strictly increasing the value of the allocation.
Thus, setting $p_1 \neq p_2$ is never optimal.
We restrict our attention to prices $\underline{p},\overline{p}$ for both advertisers, and we search for the best allocation and prices.
We show below that the agents cannot increase their utility by using  prices different from $\underline{p},\overline{p}$.
The possible combinations of prices are the following:
\begin{itemize}
\item $(p_1= \underline{p}, \,p_2 = \underline{p})$: one of best allocations is $f(1) = 1$, $f(2)=2$, and the social welfare of the allocation is $2\, \underline{p}$;
\item $(p_1 = \overline{p}, \,p_2 = \overline{p})$: the best allocation is $f(1) = 1$, $f(2)=2$, and the social welfare of the allocation is $(1 + \delta)\, \overline{p}$.
%
%
%
\end{itemize}
Thus, the best allocation is for $(p_1 = \overline{p}, \,p_2 = \overline{p})$ as, by construction, $2 \, \underline{p}< \overline{p} < (1+\delta)\,\overline{p}$.
Now, we show that no price different than $\underline{p},\overline{p}$ can lead to a better social welfare.
Notice that any price $< \underline{p}$ leads to a value strictly smaller than that provided by $\underline{p}$ as $q_1 = q_2 = 1$, and for any price $> \overline{p}$ we have $q_1,q_2 = 0$ and therefore the social welfare is zero.
Thus, we can safely restrict our attention to the set $(\underline{p},\overline{p})$.
By construction of $\psi(p_2)$, for every $\underline{p} < p_2 < \overline{p}$ we have:
\begin{align*}
(1+\psi(p_2))\,p_2 &=  \\
\left(1+\dfrac{\left(1+\delta\right)\, \left(\overline{p} - \underline{p}\right)}{\left(1 - \frac{1+\delta}{2}\right)\,p_2 + \frac{1+\delta}{2}\,\overline{p}-\underline{p} } -1\right) \, p_2 & =  \\
\dfrac{\left(1+\delta\right)\, \left(\overline{p} - \underline{p}\right)}{\left(1 - \frac{1+\delta}{2}\right)\,p_2 + \frac{1+\delta}{2}\,\overline{p}-\underline{p} }\, p_2,
\end{align*}
which is an hyperbola whose supremum is for $p_2 \rightarrow \overline{p}$.
Thus, for every $p_2 < \overline{p}$ we have:
\[
(1+\delta)\,\overline{p} > (1+\psi(p_2))\,p_2,
\]
and therefore the optimal allocation is for $(p_1 = \overline{p}, \,p_2 = \overline{p})$.
The payments are such that $\pi_1 = \underline{p} - \delta\, \overline{p} > 0 $ (since the best allocation without ad~$1$ is $f(2)=1$ with $p_2 = \underline{p}$) and $\pi_2 = p_1 - p_1 = 0$ (since ad~$2$ does not introduce any externality in the optimal allocation).
Therefore, the revenue of the mechanism is $\underline{p} - \delta\, \overline{p} > 0 $.
\textsc{Revenue with \ivcg}.
Initially, we observe that when $p_1 = p_2$, the payments are zero for every pair of declared gain $b_1,b_2$. 
Basically, this is because the prices are set by the agents, and therefore the price used in the optimal allocation with both ads and the price used by the VCG payments for the optimal allocation when an ad is removed from the market are the same. 
Since those prices are the same and $\lambda_1 = \lambda_2$, the values of an ad in the optimal allocation with both ads and in the allocation when the other ad is discarded are the same.
Therefore, a strictly positive revenue of the mechanism is possible only when $p_1 \neq p_2$.
In the following, we show that there is no Nash equilibrium with $p_1 \neq p_2$.
Consider any input profile $(b_1,p_1,b_2,p_2)$ in which $p_1 \neq p_2$.
As argued above, when $p_1 \neq p_2$, no more than one ad with strictly positive value can be allocated due to the definition of functions $q_1,q_2$.
Assume that ad $i$ is allocated, while ad $-i$ is not.
Focus on the case $p_i\leq \overline{p}$.
We have $u_{-i} = 0$, ad~$-i$ not being allocated. 
If agent~$-i$ inputs $p_{-i} = p_i$ and any $b_{-i}>0$, then she gets a utility $u_{-i} = q_{-i}(p_i,p_i)\,p_i >0$ as agent~$i$'s payment is $\pi_{-i}=p_i - p_i = 0$.
Therefore, agent~$-i$ would deviate to play $p_{-i} = p_i$ and any $b_{-i}>0$.
Notice that inputting these values may be not the best response of agent~$-i$. 
However, it is always the case that, if $p_i \neq p_{-i}$ and ad~$-i$ is not allocated, then agent~$-i$ can increase her utility by playing $p_{-i} = p_i$.
Focus on the case $p_i> \overline{p}$.
In this case, ad~$i$ gets no value from being allocated, and therefore she will reduce the price such that $p_i\leq \overline{p}$.
Thus, in the setting provided above, there is no Nash equilibrium in which $p_1 \neq p_2$.
Finally, we show that the inputs $(b_1= \overline{p},p_1= \overline{p},b_2= \overline{p},p_2= \overline{p})$ are in equilibrium.
We notice that for any $b_1,b_2>0$, both ads are displayed and the value of the allocations keeps to be $(1+\delta)\,\overline{p}$.
Thus, we analyze possible deviations to different values of $p_1=p_2= \overline{p}$.
Notice that such deviations would lead the agents to have different values of prices and therefore only one ad is displayed.
Focus on agent~$1$.
Any $p_1 > \overline{p}$ would make ad~$1$ not be allocated.
Any $p_1 < \overline{p}$ making ad~$1$ be the only allocated ad would give agent~$1$ a value $< \overline{p}$ and would charge agent~$1$ of $\delta\,\overline{p}$, leading to a utility $u_1 < \overline{p}$.
Thus, agent~$1$ cannot improve her utility by deviating from $p_1 = \overline{p}$.
Focus on agent~$2$.
Any $p_2 \neq \overline{p}$ would make ad~$2$ not be allocated.
Thus, agent~$2$ cannot improve her utility by deviating from $p_2 = \overline{p}$.
This means that there is always a pure-strategy Nash equilibrium in which $p_1 = p_2$. 
Notice that we do not exclude the case in which there are other Nash equilibria than $(b_1= \overline{p},p_1= \overline{p},b_2= \overline{p},p_2= \overline{p})$.
However, any other Nash equilibrium is with $p_1= p_2$, thus providing a revenue of zero to the mechanism.
This shows that the Price of Stability is unbounded in \ivcg.
\textsc{Revenue with \igsp}.
The proof in this case follows arguments similar to those used above for the case of \ivcg.
In order to guarantee the existence of a Nash equilibrium, we need to assume that ad~$i$ is allocated even if she declares a gain $b_i=0$.
%
%
Initially, we observe that, when $p_1 = p_2$, there is no Nash equilibrium in which both $b_1$ and $b_2$ are strictly larger than $0$. 
Assume by contradiction that both $b_1$ and $b_2$ are strictly larger than $0$ and $p_1 = p_2$.
Assume, w.l.o.g., that $f(i)= 1$ and $f(-i)= 2$.
Agent~$i$'s payment is $\pi_i = b_{-i}> 0$, while agent~$-i$'s payment is $\pi_{-i} = 0$.
Thus, agent~$i$ can improve her utility by bidding $b_i= 0$ such that ad~$i$ is displayed in the second slot so as to be charged of a payment $\pi_i = 0$. 
Therefore, we have a contradiction and $b_1$ and $b_2$ cannot be both strictly larger than $0$ in a Nash equilibrium when $p_1 = p_2$.
%
%
Notice that, if ad~$i$ is not displayed if $b_i=0$, every agent's best response would be that of making a strictly positive bid smaller than the opponent's bid and therefore the best response dynamics will never reach a fixed point, thus leading to the non-existence of the equilibrium. 
Notice also that, if $b_i = 0$, then $f(i)=2$ and the sum of the payments is such that $\pi_i+\pi_{-i}=0$.
As discussed for the case of the indirect-revelation VCG mechanism, if prices $p_1,p_2$ are different, the non-allocated ad~$i$ can improve her utility by changing her price $p_i$ as $p_i = p_{-i}$.
Furthermore, as above, setting just $p_i = p_{-i}$ may be not the best response of agent~$i$, but there is no Nash equilibrium when $p_i \neq p_{-i}$.
Finally, we prove that there is at last a Nash equilibrium when $p_1 = p_2$. 
In particular, a Nash equilibrium is $(b_1 = \overline{p}, p_1 = \overline{p}, b_2 = 0, p_2 = \overline{p})$.
As remarked above, both payments $\pi_1,\pi_2$ are zero as $p_1 = p_2$.
Furthermore, as in the case of \ivcg, deviating toward a price different from $\overline{p}$ would make that only a single ad has a strictly positive value from being allocated. 
If agent~$i$ deviates to $p_i < \overline{p}$ and it is the ad that is displayed, then her value is smaller than that she gets when $\overline{p}$.
Thus, agent~$i$ has no strictly positive incentive to deviate.
This shows that $(b_1 = \overline{p}, p_1 = \overline{p}, b_2 = 0, p_2 = \overline{p})$ is a Nash equilibrium and that the revenue of the mechanism is zero.
As in the case of the indirect-revelation VCG mechanism, there is no guarantee that such a Nash equilibrium is unique.
However, any other Nash equilibrium is with $p_1=p_2$ and $b_i= 0$ for at least one agent~$i$.
Thus, every Nash equilibrium provides a revenue of zero to the mechanism.
%
%
\end{proof}

\thmeleven*

\begin{proof}
 It follows from two considerations:
\begin{itemize}
\item with VCG payments, bidding the real $b_i$ and the price $p_i$ that \dvcg{} would choose is a Nash equilibrium, and
\item with $m=1$ slot, the payments of \dvcg{} and \ivcg{} are the same.
\end{itemize}
This concludes the proof.
\end{proof}

\thmtwelve*

\begin{proof}
The proof is based on the following setting.
\textsc{Setting}. Consider the following setting:
\begin{itemize}
\item $n=2$;
\item $m=1$;
\item $q_1(p_1,p_{\min}) = 
\begin{cases}
1	& 	\text{if } p_1 \leq \overline{p} \text{ and } p_1 = p_{\min}	\\
0& 	\text{otherwise}
\end{cases}$;
\item $q_2(p_2,p_{\min}) = 
\begin{cases}
1		& 	\text{if } p_2 \leq \underline{p}\\
0		& 	\text{otherwise}
\end{cases}$;
\item $c_1 = c_2 = 0$;
\item $\alpha_1 = \alpha_2 = 1$;
\item $\lambda_1 =1$,
\end{itemize}
where $0<\underline{p}< 0.5\, \overline{p}.$ 

\textsc{Revenue with \dvcg}. 
In The optimal allocation we have $f(1)=1$ and $f(2)=\bot$ with $p_1 = \overline{p}$ and $p_2 = \underline{p}$. The revenue $\Rev$ is $\underline{p}>0$.

\textsc{Revenue with \igsp}. 
Every Nash equilibrium prescribes that $b_1>\underline{p}\geq b_2$ and $p_1=\overline{p}$ as agent~$1$ can get the first slot, payment $\pi_1$ is zero as $p_1 =p_{\min}>p_2$, and the expected value of agent~$1$ is maximized.
Furthermore, payment $\pi_2 = 0$ since ad~$2$ is not allocated.
Thus, revenue $\Rev$ is zero leading to a PoS of $\infty$.
\end{proof}

\ivcgstarthm*

\begin{proof}
	First we show that the mechanism is individually rational. In particular, we show that $\pi_i\le \lambda_{\tilde f(i)}\,q_i(p_i,p_{\min}) \,b_i$ for every $i \in N$, where $\tilde f$ is the allocation chosen by the mechanism \ivcgstar{}.
	We distinguish two cases. If $\widehat{\SW} \geq \max_j \overline{\SW}^{-j}$, then the allocation returned by the mechanism is $\bar f$ and the payment  $\pi_i=\overline{\SW}^{-i}-\widehat{\SW}+\lambda_{\bar f(i)}q_i(p_i,p_{\min}) b_i\le \lambda_{\bar f(i)}q_i(p_i,p_{\min}) b_i$, since $\widehat{\SW} \geq  \overline{\SW}^{-i}$.
	If $\widehat{\SW} \geq \max_j \overline{\SW}^{-j}$, than $\pi_i
	=0$ by construction. Therefore, $u_i$ is always non-negative.
	Moreover, the mechanism is weakly budget-balanced since all payments are at least  $0$.
	In particular, if $\widehat{\SW} \geq \max_i \overline{\SW}^{-i}$, the payment of agent $i$ is $\pi_i=\overline{\SW}^{-i}-\widehat{\SW}+\lambda_{\bar f(i)}q_i(p_i,p_{\min}) b_i\ge 0$ by the optimality of $\overline{\SW}^{-i}$. Otherwise, all the payments are $0$.
	
	Finally, we show that bidding truthfully and with the prices that \dvcg{} would charge  is an equilibrium with the same revenue of \dvcg.
	Let $\alpha_i$ and $c_i$ be the private information of agent $i$. 
	Moreover, for each agent $i$, let $p_i$ be the price selected by \dvcg{} with truthful bidding, $b_i=\alpha_i (p_i- c_i)$, and $p^*_i=\arg \max_p \alpha_i q(p,p) (p-c_i)$.
	Since $p^*_i$ maximizes $ \alpha_i q(p,p) (p- c_i)$, then its derivative is $0$ in $p^*_i$, \emph{i.e.}, $\frac{dq(p_i,p_i)}{dp_i}\big\vert_{p_i=p_i^*} (p^*_i-c_i)+q(p^*_i,p^*_i)=0$, implying that $c_i= q(p^*_i,p^*_i)/\frac{dq(p_i,p_i)}{dp_i}\big\vert_{p_i=p_i^*}+p^*_i$.
	This requires that the derivative of $q_i(p_i^*,p_i^*)$ is strictly positive in $p_i^*$.
	Hence, \ivcgstar correctly computes $\hat c_i=c_i$.
	Moreover, since $b_i=\alpha_i (p_i- c_i)$, then $\alpha_i=\frac{b_i}{p_i-c_i}$ and $\hat \alpha_i=\alpha_i$. \footnote{Notice that we can assume that in the optimal allocation $p_i-c_i\neq 0$. Otherwise the utility of the agent is $0$ and there exists an allocation with the same $\SW$ that does not assign any slot to the agent and increases the price and the expected gain.}
	Hence, since the agents submit the optimal prices, the mechanism computes a $\widehat{\SW}$ that is equivalent to the social welfare of \dvcg.
	Moreover, also $\overline{\SW}^{-i}$ is the same of the direct mechanism for each $i$, since it optimizes over the real bidders' types. Hence, $\widehat{\SW} \ge \max_i \overline{\SW}^{-i}$ and the mechanism has the same payments and revenue of the direct mechanism.
	We conclude the proof showing that this is an equilibrium.
	Similarly to the proof of Theorem \ref{thm:indirect}, we have:
	$$
	u_i(\bar f, \p, \bpi) = \widehat{\SW} - \overline{\SW}^{-i}\ge 0$$
	Two cases are possible. If agent $i$ changes the strategy keeping $\widehat{\SW} \ge \overline{\SW}^{-i}$, the claim follows by the optimality of $\widehat{\SW}$.
	Otherwise, the utilities of all the agents are $0$.
	This concludes the proof.
\end{proof}

\end{document}